\documentclass[supp]{new_tlp} 
\usepackage{times}
\usepackage{helvet}
\usepackage{courier}
\usepackage{color}
\usepackage{graphics}
\usepackage{epsf}
\usepackage{rotating}
\usepackage{times}
\usepackage{amsmath,  amssymb}
\usepackage{graphicx}
\usepackage{epsfig}
\usepackage{amsfonts}

\usepackage{comment}
\usepackage{supp}
\long\def\comment#1{}

\newtheorem{theorem}{Theorem}

\newtheorem{definition}{Definition}

\newtheorem{lemma}{Lemma}
\newtheorem{example}{Example}

\newcommand{\HB}{\textit{HB}}
\newcommand{\KB}{\textit{KB}}
\newcommand{\FP}{\textit{FP}}
\newcommand{\lfp}{\textit{lfp}}

\newcommand{\WFS}{\textit{WFS}}

\newcommand{\Effective}{\textit{Effective}}

\renewcommand{\Pi}{\varPi}
\renewcommand{\Omega}{\varOmega}
\renewcommand{\Psi}{\varPsi}
\renewcommand{\Phi}{\varPhi}

\usepackage{comment}
\long\def\comment#1{}

  \title[Theory and Practice of Logic Programming]
        {A Well-Founded Semantics for FOL-Programs}

  \author[Y. Bi.  J. You.  Z. Feng.]
         {Yi Bi$^1$, Jia-Huai You$^2$, Zhiyong Feng$^1$
         \\
          $^1$Tianjin University, Tianjin China\\
           $^2$University of Alberta, Edmonton T6G 2E8, Canada
         }

\jdate{March 2003}
\pubauthor{Bi}
\pubyear{2003}
\pagerange{\pageref{firstpage}--\pageref{lastpage}}
\doi{S1471068401001193}

\begin{document}

\label{firstpage}
\maketitle
\begin{abstract}
\begin{quote}
An FOL-program consists of a background theory in a decidable fragment of first-order logic and a collection of rules possibly containing first-order formulas. The formalism stems from recent approaches to tight integrations of ASP with description logics. In this paper, we define a well-founded semantics for FOL-programs based on a new notion of unfounded sets on consistent as well as inconsistent sets of literals, and study some of its properties. The semantics is defined for all FOL-programs, including those where it is necessary to represent inconsistencies explicitly. The semantics supports a form of combined reasoning by rules under closed world as well as open world assumptions, and it is a generalization of the standard well-founded semantics for normal logic programs. We also show that the well-founded semantics defined here approximates the well-supported answer set semantics for normal DL programs.

To appear in Theory and Practice of Logic Programming (TPLP). 
\end{quote}
\end{abstract}


\section{Introduction}
Recent literature has shown extensive interests in combining ASP with fragments of classical logic, such as description logics (DLs) (see, e.g.,  \cite{BruijnET08,Bruijn-2007,EiterILST08,lukas-IEEE,Motik:JACM:2010,Rosati05,Rosati06,ShenW11,YangYF11}).
A program in this context is a combined knowledge base $\KB = (L, \Pi)$, where $L$ is a knowledge base of a decidable fragment of first-order logic and $\Pi$ a set of rules possibly containing
first-order formulas or interface facilities.
In this paper, we use {\em FOL-program} as an umbrella term for approaches
that allow first-order formulas to appear in rules (the so-called {\em tight} integration), for generality. The goal of this paper is to formulate a well-founded semantics for these programs with the following features.
\begin{itemize}
  \item The class of all FOL-programs are supported.
  \item Combined reasoning with closed world as well as open world assumptions is supported.
\end{itemize}

Under the first feature, we shall allow an atom with its predicate shared in the first-order theory $L$ to appear in a rule head. This can result in two-way flow of information
and enable inferences within each component automatically.
For example,
assume $L$ contains a formula that says students are entitled to educational discount, $\forall x ~St(x) \supset  EdDiscount(x)$.  Using the notation of DL, we would write $St \sqsubseteq EdDiscount$.
Suppose in an application anyone who is not employed full time but registered for an evening class is given the benefit of a student.
We can write a rule
$$St(X) \leftarrow  EveningClass(X), not~ HasJob(X).$$
Thus, that such a person enjoys educational discount can be inferred directly from the underlying knowledge base $L$.

To support all FOL-programs, we need to consider the possibility of inconsistencies arising in the construction of the intended well-founded semantics. For example, consider an FOL-program, $\KB=(L,\Pi)$, where
$L = \{\forall x A(x) \supset C(x), \neg C(a)\}$ and $\Pi = \{A(a) \leftarrow not~B(a); \\
B(a) \leftarrow B(a) \}$.  Suppose the Herbrand base is $\{A(a),B(a)\}$.
In an attempt to compute the well-founded semantics of $\KB$ by an iterative process, we begin with  the empty set; while $L$ entails $\neg A(a)$,
since
$B(a)$ is false by closed world reasoning, we derive $A(a)$ resulting in an
inconsistency.  This reasoning process suggests  that during an iterative process  a consistent set of literals may be mapped to an inconsistent one and, in general, whether  inconsistencies arise or not is not
known {\em a priori} without actually performing the computation.

That the well-founded semantics of an FOL-program is defined by an inconsistent set can be useful on its own, or in the computation of (suitably defined) answer sets of the program.  If we have computed the well-founded semantics which is inconsistent, we need not pursue the task of computing answer sets of the same program, because they do not exist.

In complex real world reasoning by rules, it is sometimes desirable that not all predicates are reasoned with under the closed world assumption.  Some conditions may need to be established under the open world assumption. We call this {\em combined reasoning}.
For example, we may write a rule
$$PrescribeTo(X, Q) \leftarrow \Effective(X, Z), Contract(Q, Z), \neg AllergicTo(Q, X)$$
to describe that an antibiotic is prescribed to a patient who contracted a bacterium,
if the antibiotic against that bacterium is effective and patient is not allergic to it.  Though $\Effective$ can be reasoned with under the closed world assumption,  it may be  preferred to judge whether a patient is not allergic to an antibiotic under the open world assumption, e.g., it holds if it can be proved classically. This is in contrast with closed world reasoning whereas one may infer nonallergic due to lack of evidence for allergy.

To our knowledge, there has been no well-founded semantics defined for all  FOL-programs.  The closest that one can find is the definition for a subset of FOL-programs \cite{lukas-IEEE}, which relies on syntactic restrictions so that the least fixpoint is computed over consistent sets of literals. To ensure that the construction is well-defined,
 it is assumed that DL axioms must be, or can be converted to, {\em tuple generating dependencies} (which are essentially Horn rules) plus constraints.  Thus, the approach cannot be lifted to handle first-order formulas in general. In addition, to the best of our knowledge, no combined reasoning is ever supported under any well-founded semantics.

As motivated above, in this paper we
first define a well-founded semantics for FOL-programs based on a new notion of unfounded sets. We show that the semantics generalizes the well-founded semantics for normal logic programs.  Also, we prove that the well-founded semantics defined here approximates the well-supported answer set semantics for the language
of \cite{ShenW11}; namely, all well-founded atoms (resp. unfounded atoms) of a program remain to be true (resp. false) in any well-supported answer set. This makes it possible to use the mechanism of constructing the well-founded semantics as constraint propagation in an implementation of  computing well-supported answer sets.

The paper is organized as follows.
The next section introduces the language and notations.  In Section \ref{WFS} we define the well-founded semantics. Section \ref{p} studies some properties and relates to the well-supported answer set semantics, followed by related work in Section \ref{r}.  Section \ref{f} concludes the paper and points to future directions.

\section{Language and Notation}

We assume a language of a decidable fragment of first-order logic, denoted ${\cal L}_\Sigma$, where $\Sigma= \langle F^n; P^n\rangle$,  called a signature,  and $F^n$ and $P^n$ are disjoint countable sets of $n$-ary function and $n$-ary predicate symbols, respectively. Constants are 0-ary functions. {\em Terms} are variables, constants, or functions in the form $f(t_1,...,t_n)$, where each $t_i$ is a term and $f \in F^n$.  {\em First-order formulas}, or just {\em formulas}, are defined as usual, so are the notions of {\em satisfaction}, {\em model}, and {\em entailment}.

Let $\Phi_P$ be a finite set of predicate symbols and $\Phi_C$ a nonempty finite set of constants such that $\Phi_C \subseteq F^n$.
An {\em atom} is of the form $P(t_1,...,t_n)$ where $P\in \Phi_P$
and each $t_i$ is either a constant from $\Phi_C$ or a variable. A {\em negated atom} is of the form $\neg A$ where $A$ is an atom. We do not assume any other restriction on the vocabularies, that is, $\Phi_P$ and $P^n$ may have predicate symbols in common.



An {\em FOL-program}
is a combined knowledge base $\KB = (L,\Pi)$,
where $L$ is a first-order theory of ${\cal L}_\Sigma$
and $\Pi$ a {\em rule base}, which is a finite collection of rules of  the form
\begin{eqnarray}
\label{rule}
H \leftarrow A_1,\ldots, A_m, not ~B_1, \ldots, not ~B_n
\end{eqnarray}
where $H$ is an atom, and $A_i$ and $B_i$ are atoms or formulas. By abuse of terminology, each $A_i$ is called a {\em positive literal} and each $not~B_i$ is called a {\em negative literal}.


For any rule $r$, we denote by $head(r)$ the head of the rule and $body(r)$ its body, and we define $pos(r) = \{A_1, ..., A_m\}$ and $neg(r) = \{B_1, ..., B_n\}$.

A {\em ground instance} of a rule $r$ in $\Pi$ is obtained by replacing every free variable with a constant from $\Phi_C$. In this paper, we assume that a rule base $\Pi$ is already grounded if not said otherwise.
When we refer to an atom/literal/formula, by default we mean it is a ground one.

Given an FOL-program $\KB = (L, \Pi)$, the {\em Herbrand base} of $\Pi$, denoted $\HB_\Pi$, is the set of all ground atoms
$P(t_1,...,t_n)$, where $P\in \Phi_P$
occurs in $\KB$ and $t_i \in \Phi_C$.

We denote by $\Omega$ the set of all predicate symbols appearing in $\HB_\Pi$ such that $\Omega \subseteq P^n$.  For distinction, we call atoms whose predicate symbols are not in $\Omega$ {\em ordinary}, and all the other formulas {\em FOL-formulas}.  If $L =\emptyset$ and $\Pi$ only contains rules of the form (\ref{rule}) where all $H$, $A_i$ and $B_j$ are ordinary atoms, then $\KB$ is called a {\em normal logic program}.

Any subset $I \subseteq \HB_\Pi$ is called an {\em interpretation} of $\Pi$.  It is also called a {\em total} interpretation or a {\em 2-valued} interpretation.
If $I$ is an interpretation, we define
$\bar{I}= \HB_\Pi \backslash I$.

Let  $Q$ be a set of atoms. We define $\neg. Q = \{\neg A~|~ A \in Q\}$.
For a set of atoms and negated atoms $S$, we
define
$S^+ = \{ A \,|\, A \in S\}$, $S^- =\{ A \,|\, \neg A \in S\}$, and
$S|_\Omega = \{ A \in  S~|~pred(A) \in \Omega\}$, where $pred(A)$ is the predicate symbol of $A$.
Let $Lit_\Pi = \HB_\Pi \cup \neg.\HB_\Pi$.
A subset $S \subseteq Lit_\Pi$ is consistent if $S^+ \cap S^- = \emptyset$.
For a first-order theory $L$, we say that $S \subseteq Lit_\Pi$ is {\em consistent with $L$} if the first-order theory $L\cup S|_\Omega$ is consistent (i.e., the theory is {\em satisfiable}).
Note that when we say $S$ is consistent with $L$, both $S$ and $L$ must be consistent. Similarly,  a (2-valued) interpretation $I$ is consistent with $L$ if $L \cup I|_\Omega \cup \neg.\bar{I}|_\Omega$ is consistent. We denote by $Lit_\Pi^c$ the set of all consistent subsets of $Lit_\Pi$. For any $S \in Lit_\Pi^c$, $S'$ is called a {\em consistent extension} of $S$ if $S \subseteq S' \in Lit_\Pi^c$.

\begin{definition}
\label{2-valued}
  Let $\KB = (L,\Pi)$ be an FOL-program and $I \subseteq \HB_\Pi$ an interpretation. Define
the satisfaction relation under $L$, denoted $\models_L$, as follows (the definition extends to conjunctions of literals):
\begin{enumerate}
  \item For any ordinary atom $A \in \HB_\Pi$, $I\models_L A$ if $A\in I$ and $I \models not ~A$ if $A \not \in I$.
  \item For any FOL-formula $A$, $I\models_L A$ if $L \cup I|_\Omega \cup \neg.\bar{I}|_\Omega \models A$, and
  $I  \models_L not~A$ if $I \not \models_L  A$.
\end{enumerate}
\end{definition}

Let $\KB=(L,\Pi)$ be an FOL-program.
For any  $r \in \Pi$ and $I \subseteq \HB_\Pi$,
$I \models_L r$ if $I \not\models_L body(r)$ or $I \models_L head(r)$.
$I$ is a {\em model} of $\KB$ if $I$ is consistent with $L$ and $I$ satisfies all rules in $\Pi$.


\begin{example}
\label{1}
To illustrate the flexibility provided by the parameter $\Omega$, suppose we have a program $\KB = (L,\Pi)$ where $\Pi$ contains a rule that says any unemployed with disability receives financial assistance, with an FOL-formula in the body
$$Assist(X) \leftarrow Disabled(X), not~Employed(X)$$
Assume $\Omega = \Phi_P = \{Assist, Employed\}$. Then,  $Employed$ is interpreted under the closed world assumption
and $Disabled$ under the open world assumption. Indeed, unemployment can be established by closed world reasoning for lack of evidence of employment, but
disability requires a direct proof.
\end{example}

\section{Well-Founded Semantics}
\label{WFS}

We first define the notion of unfounded set.
Intuitively, the atoms in an unfounded set can be safely assigned to false, due to persistent inability to derive their positive counterparts.

\begin{definition}
\label{u}
{\bf (Unfounded set)} Let $\KB = (L,\Pi)$ be an FOL-program and $I \subseteq Lit_\Pi$. If $I \cup L$  is consistent, then a set $U \subseteq \HB_\Pi$ is an {\em unfounded set} of $\KB$ relative to $I$ iff
for every $H\in U$ and $r\in \Pi$, both of the following conditions are satisfied
\begin{itemize}
\item [(a)] If $head(r) = H$, then
\vspace{-0.1in}
  \begin{itemize}
  \item[(i)] $\neg A \in I \cup \neg .U$ for some ordinary atom $A\in pos(r)$, or
  \item[(ii)] $B\in I$ for some ordinary atom $B\in neg(r)$, or
  \item[(iii)] for some FOL-formula $A\in pos(r)$, it holds that $L\cup S|_\Omega \not \models A$,
for all $S \in Lit_\Pi^c$ with $I\cup \neg.U\subseteq S$, or
  \item[(iv)] for some FOL-formula $A\in neg(r)$, $L  \cup   I |_\Omega  \models A$.
\end{itemize}
 \item[(b)]
$L \cup S|_\Omega \not \models H$ for all $S\in Lit_\Pi^c$ with $I\cup \neg.U\subseteq S$.
\end{itemize}
If $I \cup L$ is inconsistent, the unfounded set of $\KB$ relative to $I$ is $\HB_\Pi$.
\end{definition}


That $H$ is unfounded relative to $I$  if both conditions (a) and (b) are satisfied when $I \cup L$ is consistent;
 in particular, condition (a.iii) ensures that a positive occurrence of an FOL-formula in the rule body is not entailed, for all consistent extensions of $I \cup \neg . U$; and
condition (b) ensures the inability to infer its positive counterpart, independent of any rules.

An FOL-formula may contain shared predicates in $\Omega$, and those not in $\Omega$ hence not shared. The latter are supposed to be interpreted under  open world assumption. Continuing with Example \ref{1} above,  let $\KB= (L, \Pi)$, where
$$
\begin{array}{ll}
L = \{\forall x ~Certified(x) \supset Disabled(x)\} \\
\Pi = \{Assist(a) \leftarrow Disabled(a), not~Employed(a)\}
\end{array}
$$
Assume $Assist, Employed \in \Omega$ while
$Certified$ and $Disabled$ are not.
Let $\Phi_C = \{a\}$, and thus $\HB_\Pi = \{Assist(a), Employed(a)\}$.
Clearly, $\{Assist(a), Employed(a)\}$ is an unfounded set relative to $I = \emptyset$, in particular because $Disabled(a)$ is not derivable under all consistent extensions of $I$.  Note that, since $Disabled(a)$ is not in $\HB_\Pi$, it is not part of an unfounded set.\footnote{Placed under the context of 2-valued logic, the reasoning here is analogue to parallel
circumscription \cite{McCarthy-circ}, where
the predicates $Employed$ and $Assist$ are minimized with predicates $Certified$ and $Disabled$ varying.
}

\begin{lemma}
Let $\KB = (L,\Pi)$ be an FOL-program and $I \subseteq Lit_\Pi$.
A set of unfounded sets of $\KB$ relative to $I$ is closed under union, and
the greatest unfounded set of $\KB$ relative to $I$ exists, which is the union of all unfounded sets of $\KB$ relative to $I$.
\end{lemma}
\begin{proof}
If $I$ is inconsistent, the claims hold trivially. For a consistent $I$,
suppose $U_1$, $U_2\subseteq \HB_\Pi$ are unfounded sets (of $\KB$ relative to $I$), we show that $U_1 \cup U_2$ is also an unfounded set.
Let $A\in U_1$. Since both (a) and (b) in Definition \ref{u} hold for $U_1$ and $U_2$ separately, in particular, each consistent extension of $I \cup \neg . (U_1 \cup U_2)$ is a consistent extension of $I \cup \neg . U_1$,
(a) and (b) also hold for $U_1 \cup U_2$. Thus $A \in U_1 \cup U_2$.
By symmetry,
the same argument applies to $U_2$.
Therefore, the union of
all unfounded sets is an unfounded set, which is the greatest among all unfounded sets.\end{proof}

We define the operators which will be used to define the well-founded semantics.

\begin{definition}
\label{7}
Let $\KB = (L,\Pi)$ be an FOL-program. Define $T_\KB$, $U_\KB$, $Z_\KB$
as mappings of $2^{Lit_\Pi} \rightarrow \HB_\Pi$,
and $W_\KB$ as a mapping of $2^{Lit_\Pi} \rightarrow 2^{Lit_\Pi}$, as follows:
\begin{itemize}
  \item[(i)] If $I \cup L$ is inconsistent, then $T_\KB (I) \!=\! \HB_\Pi$; otherwise,
$H  \in  T_\KB (I)$ iff $H  \in \HB_\Pi$ and
either (a) or (b) below holds
  \begin{itemize}
  \item[(a)]  some $r\in \Pi$ with $head(r) = H$ exists such that
\vspace{-.02in}
    \begin{itemize}
      \item[(1)]  for any ordinary atom $A$, $A \in I$ if $A \in pos(r)$ and $\neg A \in I$ if $A \in neg(r)$,
\item[(2)]
for any FOL-formula $A\in pos(r)$, $L \cup I|_\Omega \models A$, and
      \item[(3)]  for any FOL-formula $B\in neg(r)$, $L\cup S|_\Omega \not \models B$, for all  $S \in Lit_\Pi^c$ with $I\subseteq S$.
    \end{itemize}
  \item[(b)] $L \cup I|_\Omega \models H$.
  \end{itemize}
  \item[(ii)]$U_\KB (I)$ is the greatest unfounded set of $\KB$ relative to $I$.
\item
[(iii)] $Z_\KB(I) =  \{ A \in \HB_\Pi~|~ L \cup I |_\Omega \models \neg A\}$.
\item[(iv)]
$W_\KB (I) = T_\KB (I)\cup \neg.U_\KB (I) \cup \neg . Z_\KB (I)$.
\end{itemize}
\end{definition}

The operator $T_\KB$ is a consequence operator.
An atom is a consequence,  either due to a derivation via a rule (case (i.a)), or because it is entailed by $L$, given $I$ (case (i.b)). In the first case, the body of such a rule should be satisfied not only by $I$, but
by all consistent extensions of $I$. In
the case (i.a.1)  or (i.a.2), it is sufficient that the body is satisfied by $I$ only  because  the classical entailment relation is monotonic. For the  case (i.a.3) the condition needs to  be stated explicitly.

There are two features in this definition that are non-conventional. The first is
the operator $Z_\KB$ \-- interacting an FOL knowledge base with rules may result in {\em direct negative consequences}. In the second, all operators here are defined on all subsets of $Lit_\Pi$, including inconsistent ones.\footnote{When inconsistency arises, the fixpoint operator here leads to triviality. This is the most common treatment when
the underlying entailment relation is the classical one.
However, we remark that this is only one possible choice.
}

\begin{lemma}
The operators $T_\KB$, $U_\KB$, $Z_\KB$,
and $W_\KB$ are all monotonic.
\end{lemma}
\begin{proof}
Let $I_1 \subseteq I_2 \subseteq Lit_\Pi$ and $H \in T_\KB (I_1)$. Since
 any condition in part (i) of Definition \ref{7} that applies under $I_1$ applies under $I_2$ (including the case where one of them, or both, are  inconsistent with $L$),
thus the set of all consistent extensions of $I_2$ is a
subset of all consistent extensions of $I_1$, and therefore we have $T_\KB(I_1) \subseteq T_\KB(I_2)$.
The same argument applies to $U_\KB$ and $Z_\KB$.
Since $T_\KB$, $U_\KB$, and $Z_\KB$ are monotonic, it follows from the definition that the operator $W_\KB$ is also monotonic.
\end{proof}

As $W_\KB$ is monotone on the complete lattice $\langle 2^{Lit_\Pi}, \subseteq \rangle$, according to  the Knaster\--Tarski fixpoint theorem \cite{Tarski1955},
its least fixpoint, $\lfp(W_\KB)$, exists.

\begin{definition}
\label{wfs}
Let $\KB = (L,\Pi)$ be an FOL-program. The {\em well-founded semantics}  of $\KB$ (relative to $\Omega$) 
 is defined by $\lfp(W_\KB)$.
\end{definition}

We allow the well-founded semantics of an FOL-program to be defined by an inconsistent set, independent of how the semantics may be used.  This may be utilized in the computation of answer sets. Suppose under a suitable definition of answer sets for an FOL-program $\KB$, $\lfp(W_\KB)$ approximates all answer sets of $\KB$.\footnote{We will show later in this paper that the well-supported answer sets of \cite{ShenW11} fall into this category.}
 If the computed $\lfp(W_\KB)$ is inconsistent then we know $\KB$ has no answer sets.


\begin{example}
Let $\KB = (\{\neg A(a)\}, \Pi)$ where $\Pi =
\{A(a) \leftarrow not~B(a), B(a) \leftarrow B(a)\}$.
Let $\Omega = \Phi_P =  \{A,B\}$ and $\Phi_C = \{a\}$.
$\lfp(W_\KB)$ is constructed as follows   (where $W_\KB^0 = \emptyset$ and
$W_\KB^{i+1} = W_\KB(W^i_\KB)$, for all $i \geq 0$):
$$
\begin{array}{ll}
W_\KB^0 = \emptyset,\\
W_\KB^1 = \{\neg A(a), \neg B(a)\},\\
W_\KB^2 = \{\neg A(a), \neg B(a), A(a)\},\\
W_\KB^3 = Lit_\Pi,\\
W_\KB^4 = W_\KB^3.
\end{array}
$$
As a result, the well-founded semantics of $\KB$ is inconsistent.
It is interesting to note that $\KB$ has a model, $\{B(a)\}$. This means that we cannot determine whether the well-founded semantics for an FOL-program is consistent or not, based on the existence of a model; when an iterative process is carried out, we have to deal with the possibility that inconsistencies may arise.
\end{example}

\begin{example}
\label{example2}
Consider $\KB = (L, \Pi)$ where $L = \{\forall x B(x) \supset A(x), \neg  A(a) \vee  C(a)\}$ and $\Pi$ consists of
$$
\begin{array}{ll}
B(a) \leftarrow B(a).\\
A(a) \leftarrow (\neg C(a) \wedge B(a)).\\
R(a) \leftarrow not~ C(a), not~A(a).
\end{array}
$$
 Let $\Phi_P = \{A, B, R\}$, $\Omega = \{A, B\}$,  and $\Phi_C = \{a\}$. Hence  $\HB_\Pi = \{A(a), B(a),R(a)\}$.
 For $I_0 = \emptyset$, we have $T_\KB (I_0) = \emptyset$, $U_\KB (I_0) = \{B(a), A(a)\}$, and $Z_\KB = \emptyset$.  $B(a)$ is in $U_\KB (I_0)$
because $B(a)$ is not derivable by any rule based on $I_0$, and $L \cup S |_\Omega \not \models B(a)$ for all $S \in Lit_\Pi^c$ with $I_0 \cup \neg . U_\KB(I_0) \subseteq S$ (condition (b) in Definition \ref{u}).  Similarly,
$A(a)$ is in $U_\KB (I_0) $. Since $C(a)$ is not derivable under all consistent extensions, we derive $R(a)$.
Therefore,
$\lfp(W_\KB) = \{\neg B(a), \neg A(a), R(a)\}$.  Note that since $C(a) \not \in \HB_\Pi$ its truth value is not part of the well-founded semantics.
\comment{
[(ii)]  Now let us place $C$ under closed world reasoning by assuming
${\bf P} = \{A, B, C, R\}$, thus $\HB_\Pi = \{A(a),$ $B(a), C(a), R(a)\}$.
Note that here $\Omega = \{A, B, C\}$. Now $U_\KB(I_0) = \{B(a), A(a), C(a)\}$, and $\lfp(W_\KB) = \{\neg B(a), \neg A(a), \neg C(a), R(a)\}$.
}
\end{example}

\comment{
\begin{proof}
$W_\KB^0 = V_\KB^0 = \emptyset$. For computing $(V_\KB^{i+1})^+$, given $V_\KB^i$, we have
$$T_\KB^0 = (V_\KB^i)^+, \cdots, T_\KB^{k+1} = T_\KB(T_\KB^k\cup \neg.(V_\KB^i)^-),\cdots,\FP_{T_\KB}(V_\KB^i)$$
We first show $\lfp(V_\KB) \subseteq \lfp(W_\KB)$, assume that, for any $i\geq0$, and for any literal $A\in V_\KB^i$, there is an $i'\geq0$ such that $A\in W_\KB^{i'}$, and next we show that for any $A\in V_\KB^{i+1}$, there is a $j\geq i'$ such that $A\in W_\KB^j$. If $A\in V_\KB^i$, by the induction hypothesis we have $A\in W_\KB^{i'}$. Otherwise, $A\in T_\KB^k$, for some $k\geq0$. Then we will have $A\in W_\KB^{i'+k}$, since if any condition in Definition \ref{7}  is satisfied by the interval $[T_\KB^k, \HB_\Pi \backslash U_\KB(V_\KB^i)]$, it remains to be satisfied by $[(W_\KB^{i'+k})^+, \HB_\Pi \backslash W_\KB^{i'+k})^-]$, as the latter is a sub-interval of the former.
\end{proof}
}

\comment{
Actually, more flexibilities are possible.  For example,
one can remove condition (b) in Definition \ref{7} and carry out the computation under it  separately.
That is, we can define yet another alternative operator, say  $V'_\KB$, as
\begin{eqnarray}
\label{V2}
V'_{\KB}(I) = \FP_{T'_{\KB}}(I) \cup  \{A ~|~L\cup I|_\Omega \models A\}
\cup \neg.U_{\KB}(I)
\end{eqnarray}
where $T'_\KB$ is the same as $T_\KB$ except condition (b) is removed.  We can show

\begin{lemma}
$\lfp(V'_\KB) = \lfp(V_\KB) = \lfp(W_\KB)$.
\end{lemma}
}

\section{Properties and Relations}
\label{p}

We first show that the well-founded semantics is a generalization of the well-founded semantics for normal logic  programs.

\begin{theorem}
Let $\KB = (\emptyset, \Pi)$ be a normal logic program.  Then, the $\WFS$ of $\KB$ coincides with the $\WFS$ of $\Pi$.
\end{theorem}
\begin{proof}
The $\WFS$ of a normal program $\Pi$ is defined by the least fixpoint of a monotone operator $W_\Pi$ on the set of consistent subsets of $\HB_\Pi \cup \neg . \HB_\Pi$:
\begin{itemize}
\item
$T_\Pi(S) = \{head(r)~|~ r \in \Pi, pos(r) \cup \neg . neg(r) \subseteq S\}$
\item
$W_\Pi(S) = T_\Pi(S) \cup \neg . U_\Pi(S)$
\end{itemize}
where $U_\Pi(S)$ is the greatest unfounded set of $\Pi$ w.r.t. $S$. A set $U \subseteq \HB_\Pi$ is an unfounded set of $\Pi$ w.r.t. $S$, if for every $a \in U$ and every rule $r \in \Pi$ with $head(r)= a$, either (i) $\neg b \in S \cup \neg . U$ for some $b \in pos(r)$, or (ii) $b \in S$ for some atom $b \in neg(r)$.

Then, it is immediate that the notion of unfounded set and the greatest unfounded set for normal logic program $\KB = (\emptyset, \Pi)$  coincide with those for $\Pi$, respectively.  Note that $Z_\KB (I) \subseteq U_\KB(I)$ when $L =\emptyset$. It is easy to see that the operator $T_\KB$ for normal program $\KB = (\emptyset, \Pi)$ reduces to $T_\Pi$ for normal program $\Pi$.
\end{proof}

The well-supported answer set semantics is defined for what are called {\em normal DL logic programs} \cite{ShenW11}, which applies to FOL-programs.
There  is however a subtle difference:  in the definition of the entailment relation, the $W_\KB$ operator uses 3-valued evaluation while the well-supported semantics is based on the notion of 2-valued {\em up to satisfaction}.

\begin{definition}
\label{up}
{\bf (Up to satisfaction)} Let $\KB = (L, \Pi)$, $l$ a literal, and $E$ and $I$ two interpretations with $E \subseteq I \subseteq \HB_\Pi$.
The relation $E$ {\em up to} $I$ {\em satisfies} $l$ {\em under} $L$, denoted $(E, I) \models_L l$,
is defined as:
$(E, I) \models_L l$ if $\forall F, E \subseteq F \subseteq I$, $F \models_L l$. The definition extends to conjunctions of literals.
\end{definition}

The entailment relation, $F \models_L l$, is based  2-valued satisfiability (cf. Def. \ref{2-valued}), i.e.,
$F\models_L l$ is
$L \cup F|_\Omega \cup \neg . {\bar F}|_\Omega\models l$, while in 3-valued satisfiability, $S \models_L l$ is  $L \cup S|_\Omega \models l$.

Given an FOL-program $\KB = (L,\Pi)$,
an immediate consequence operator is defined as:
 \begin{eqnarray}
\label{op}
  {\cal T}_\KB (E,I) = \{head(r) \mid r \in \Pi,
  (E, I) \models_L body(r)\}.
 \end{eqnarray}

The operator ${\cal T}_\KB$
 is monotonic on its first argument $E$ with $I$ fixed \cite{ShenW11}. Thus, for any model $I$ of $\KB$, we can compute a fixpoint, denoted ${\cal T}^\alpha_\KB (\emptyset, I)$.

\begin{definition}
\label{well-supported}
Let $\KB = (L,\Pi)$ be an FOL-program and $I$ a model of $\KB$. $I$ is an {\em answer set of $\KB$} if for every $A \in I$, either $A \in {\cal T}^\alpha_\KB (\emptyset, I)$ or $L \cup {\cal T}^\alpha_\KB (\emptyset, I)|_\Omega \cup \neg.\bar{I}|_\Omega \! \models \! A$.
\end{definition}

The next theorem shows that the well-founded semantics of an FOL-program approximates its well-supported answer set semantics.

\begin{theorem}
Let $\KB = (L,\Pi)$ be an FOL-program. Then every well-supported answer set of $\KB$ includes all atoms $H \in \HB_\Pi$ that are well-founded and no atoms $H \in \HB_\Pi$ that are unfounded or are in
$Z_\KB(\lfp(W_\KB))$.
\end{theorem}

\begin{proof}
To prove the assertion, it is sufficient to show that if $\lfp(W_\KB)$ is consistent, then for every well-supported answer set $I$,
all atoms in $\lfp(W_\KB)$ are in $I$ and all negated atoms in $\lfp(W_\KB)$
are in $\neg . {\bar I}$.

We consider the fixpoint construction by the operators ${\cal T}_\KB(\cdot, I)$ and  $W_\KB$. Let us use a short notation for the respective sequences by
\begin{eqnarray}
\label{9}
{\cal T}_\KB^0 = \emptyset, \ldots  , {\cal T}_\KB^{k+1} = {\cal T}_\KB ({\cal T}_\KB^k, I), \ldots ~~~~~\\
\label{10}
W_\KB^0 =\emptyset, \ldots , W_\KB^{k+1} = W(W_\KB^k), \ldots ~~~~~~
\end{eqnarray}
for all $k \geq 0$. Define $E_0 = \emptyset$ and
$E_i = \{H ~|~({\cal T}^{i-1}_\KB, I) \models_L H\}$ for all $i \geq 1$. We show that  $W^i_\KB \subseteq E_i \cup {\cal T}^{i}_\KB \cup \neg.\bar{I}$, for all $i \geq 0$. The base case is trivial. For the inductive step, assume for any $k \geq 0$ the subset relation  holds and we show it holds for $k+1$. The proof is conducted on two cases: {\bf (I)} Assume an atom $H \in W_\KB^{k+1}$ and show $H \in E_{k+1} \cup {\cal T}^{k+1}_\KB$, and {\bf (II)} Assume a negated atom $\neg H \in W_\KB^{k+1}$ and show $\neg H \in \neg.\bar{I}$.

By definition and monotonicity of the operator ${\cal T}_\KB$, $({\cal T}^{i}_\KB, I) \models_L E_i$, and it follows from the first-order entailment that for any atom $H \in \HB_\Pi$,
\begin{eqnarray}
\label{e1}
({\cal T}^{i}_\KB, I) \models_L H \mbox{ iff } (E_i\cup {\cal T}^{i}_\KB , I) \models_L H.
\end{eqnarray}
By definition,
$$
W_\KB^{k+1} = T_\KB(W_\KB^k) \cup \neg.U_\KB(W_\KB^k) \cup \neg.Z_\KB(W_\KB^k)
$$
By Proposition 1  of \cite{Shen11}, for any FOL-formula $H$,
\begin{eqnarray}
\label{p1}
(E, I) \models_L H \mbox{ iff } L \cup E |_\Omega \cup \neg.\bar{I}|_\Omega \models H.
\end{eqnarray}
By Proposition 2 \cite{Shen11}, for any ordinary atom $H$,
\begin{eqnarray}
\label{p2}
(E, I) \models_L H \mbox{ iff } H \in E; (E, I) \models_L not~H \mbox{ iff } H \not \in I.
\end{eqnarray}

{\bf (I)} For any atom $H \in W_\KB^{k+1}$, we have $H \in T_\KB(W_\KB^k)$. If condition (i.b) in Definition \ref{7} holds, we have $L \cup W_\KB^k |_\Omega \models H$. By induction hypothesis, $L \cup (E_k\cup {\cal T}^k_\KB) |_\Omega \cup \neg. \bar{I} |_\Omega \models H$. Then by (\ref{p1}) and (\ref{e1}), $({\cal T}^k_\KB, I) \models_L H$. Thus $H \in E_{k+1}$. If condition (i.a) in Definition \ref{7} holds, we consider the following four cases:
\begin{enumerate}
  \item For any ordinary atom $A \in pos(r)$, $A \in W_\KB^k$, thus $({\cal T}^k_\KB, I) \models_L A$ by (\ref{p2}) and (\ref{e1}).
  \item For any ordinary atom $A \in neg(r)$, $\neg A \in W_\KB^k$, thus $({\cal T}^k_\KB, I) \models_L not~A$.
  \item For any FOL-formula $A \in pos(r)$, $L \cup W_\KB^k|_\Omega \models A$, then $({\cal T}^k_\KB, I) \models_L A$.
  \item For any FOL-formula $A \in neg (r)$, $L \cup (W_\KB^k)'|_\Omega \not \models A$ for every $(W_\KB^k)'$ such that $W_\KB^k \subseteq (W_\KB^k)' \in Lit_\Pi^c$, we have $({\cal T}^k_\KB, I) \not \models_L A$,  since every total interpretation is a partial one.
\end{enumerate}
Hence $H \in {\cal T}^{k+1}_\KB$.

{\bf (II)} For any negated atom $\neg H \in W_\KB^{k+1}$, either $H \in U_\KB(W_\KB^k)$ or $H \in Z_\KB(W_\KB^k)$. For the case $H \in U_\KB(W_\KB^k)$, if condition (b) in Definition \ref{u} holds, we have  $({\cal T}^k_\KB, I) \not \models_L H$ by  (\ref{e1}), (\ref{p1}) and induction hypothesis, in addition to the fact that every total interpretation is a partial one. Then $H \not \in E_{k+1}$. For condition (a) in Definition \ref{u}, we also consider the following four situations:
\begin{enumerate}
  \item For any ordinary atom $A \in pos(r)$, $\neg A \in W_\KB^k \cup \neg.U$, thus $({\cal T}^k_\KB, I) \not \models_L A$, by a similar argument above.
  \item For any ordinary atom $A \in neg(r)$, $A \in W_\KB^k$, thus $({\cal T}^k_\KB, I) \models_L A$.
  \item For any FOL-formula $A \in pos(r)$, $L \cup (W_\KB^k)'|_\Omega \not \models A$, for every $W_\KB^k \subseteq (W_\KB^k)' \in Lit_\Pi^c$, then $({\cal T}^k_\KB, I) \not \models_L A$, since every total interpretation is a partial one.
  \item For any FOL-formula $A \in neg (r)$, $L \cup W_\KB^k|_\Omega \models A$, we have $({\cal T}^k_\KB, I) \models_L A$.
\end{enumerate}
We have $H \not \in {\cal T}^{k+1}_\KB$. Hence $H \not \in E_{k+1} \cup {\cal T}^{k+1}_\KB$

For any $\neg H \in W_\KB^k$, $\neg H \in \neg.\bar{I}$, since the operator $E$ and ${\cal T}_\KB$ only generate positive atoms. We then have $H \not \in E_k \cup {\cal T}^k_\KB$. As $k$ is arbitrary, we have $H \in U_\KB(\lfp(W_\KB))$ and $H \not \in E_\alpha \cup {\cal T}^\alpha_\KB$, where $E_\alpha$ and ${\cal T}^\alpha_\KB$ are the respective fixpoints. Since $E_\alpha \cup {\cal T}^\alpha_\KB = I$ (By definition \ref{well-supported}), we get $H \in \bar{I}$.
Similarly, if $H \in Z_\KB(W_\KB^k)$, then $H \in \bar{I}$. We thus have proved that $W^{k+1}_\KB \subseteq E_{k+1} \cup {\cal T}^{k+1}_\KB \cup \neg.\bar{I}$.
\end{proof}

\section{Related Work}
\label{r}

The most relevant work in defining well-founded semantics for
combing rules with DLs are \cite{Eiter:TOCL2011,lukas-IEEE}. The former embeds {\em dl-atoms} in rule bodies to serve as queries to the underlying ontology, and it does not allow the predicate in a rule head to be shared in the ontology.
In both approaches, syntactic restrictions are posted so that
the least fixpoint is always constructed over sets of consistent literals.
It is also a unique feature in our approach that combined reasoning with closed world and open world is supported.

A program in FO(ID) has a clear knowledge representation ``task" \-- the rule component is used to define concepts, whereas the FO component may assert additional properties of the defined concepts.  All formulas in FO(ID) are interpreted under closed world assumption. Thus,
FOL-programs and FO(ID) have fundamental differences in basic ideas.
On semantics, FOL-formulas can be interpreted under open world and closed world flexibly.  On modeling, the rule set in FO(ID) is built on ontologies, thus information can only flow from a  first order theory to rules. But in FOL-programs, the first order theory and rules are tightly integrated, and thus information can flow from each other bilaterally.

\section{Conclusion and Future Directions}
In this paper we have defined a new well-founded semantics for FOL-programs, where arbitrary FOL-formulas are allowed to appear in rule bodies and an atom with its predicate shared with first-order theory to appear in a rule head. Combined reasoning with closed world as well as open world is supported. Moreover, inconsistencies are dealt with explicitly, and thus the task of computing answer sets can be prejudged in case that the well-founded semantics is an inconsistent set.
We have shown that the well-founded semantics is an appropriate approximation of the well-supported answer set semantics defined in \cite{ShenW11}.

As future work, we will study the approximation fixpoint theory (AFT) \cite{denecker2000approximations,DeneckerMT04}, and investigate whether and how well-founded and stable semantics of FOL-programs can be defined uniformly under an extended approximation fixpoint theory. We are also interested in possible different approximating operators for  alternative semantics of FOL-programs. In \cite{DeneckerMT04} the authors show that the theory of consistent approximations can be applied to the entire bilattice ${\cal L}^2$ (including inconsistent elements), under the assumption that an approximating operator ${\cal A}$ is {\em symmetric}. This symmetry assumption guarantees that no transition from a consistent state to an inconsistent one may take place.
As argued at the outset of this paper, this is precisely what we cannot assume for a definition of
well-founded semantics for all FOL-programs.

\section*{Acknowledgements}
This work is supported by the National Natural Science Foundation of China (NSFC) grants 61373035 and 61373165, and by National High-tech R$\&$D Program of China (863 Program) grant 2013AA013204.

\label{f}
\comment{
We accomplish three things in this paper.  First, we define a new well-founded semantics
for FOL-programs, which also serves
as an example to demonstrate the need to deal with inconsistencies explicitly. We then extend the theory of consistent approximations of
\cite{DeneckerMT04}
to treat inconsistencies. And finally, we apply the extended theory to FOL-programs
to define well-founded and stable semantics, which are related to semantics defined using different methods.

The extended AFT developed in this paper supports two features, one for transitions from a consistent state to an inconsistent one, and the other for transitions from an inconsistent state to a possibly deeper inconsistent one. In this paper we utilize the first  to define the semantics of FOL-programs.  For potential applications of the second feature, one may consider inference systems based on non-classical logic (e.g., paraconsistent logic) or explore  the idea of belief revision. In both cases we want to infer  nontrivial conclusions  in case of an inconsistent theory. 
The $\Phi_\KB$ operator defined in this paper leaves it quite open as what may be done
in case of inconsistency, which allows further development.

In fact, inconsistencies frequently occur in Boolean networks, which have been recently related to logic programming \cite{Inoue11}. Also, there is a potential to characterize the semantics of
ASP languages that support consistency restoring rules \cite{BalaiGZ13} by the extended AFT.
In this paper, while our focus has been on the mathematical construction and reconstruction of semantics, the complexity of these semantics requires further study. These are interesting future directions.
}


\comment{
{\bf Some detailed remarks which might be useful later.}

Eiter et al. \cite{Eiter:TOCL2011} proposed a well-founded semantics for dl-programs in form of $\KB = (L,\Pi)$, where $L$ is a DL knowledge base and $\Pi$ is a collection of rules. A distinct feature is that the body of a rule may contain {\em DL-atoms}, which are well-designed interfaces to access and enhance DL knowledge base $L$.  Dl-programs differ in significant ways from normal DL logic programs. First, in a dl-program, $\Pi$ and $L$ do not share predicate symbols. Thus, information flow between $\Pi$ to $L$ relies on DL-atoms only. For normal DL logic programs, since $\Pi$ and $L$ share some predicate symbols, DL expressions are allowed to occur in $\Pi$ directly without predicate symbols mapping. Hence information flow between rule bases and ontologies bidirectionally and flexibly. Second, in a dl-program, dl-atoms are not allowed to appear in rule heads, hence we cannot derive any new knowledge about $L$ from $\Pi$. However, there is a cost for this extra feature. The formalization of the unfounded and well-founded conditions turns out to be much more complicated in our case, and there is a related computational cost.

Disjunctive dl-programs \cite{lukas-IEEE} under well-founded semantics are the closest to normal DL logic programs, if we restrict the heads of rules to be non-disjunctive. Let us call them normal dl-programs, which
differ at least in four ways from normal DL programs. Let $\KB = (L,\Pi)$ be a normal dl-program, where $\Pi$ is built over a vocabulary $\Phi = ({\bf P}, {\bf C})$, {\bf P} is a finite set of predicate symbols, and {\bf C} is a nonempty finite set of constants. (1) The approach in \cite{lukas-IEEE} only allows atomic DL expressions to occur in rule bodies and heads.  But
the normal DL logic programs allow arbitrary DL expressions occur in rule bodies.
(2) All concepts and roles occurring in $\Pi$ are required to be included in {\bf P}, so that all of them are interpreted over the Herbrand base $\HB_\Pi$ of $\Pi$. This restriction seems to prevent some atomic DL expressions in $\Pi$ to be subject to open world reasoning.
(3) In normal dl-programs, a underlying assumption is that the DL knowledge base $L$ can be cleanly separated into $L^+$ and $L^-$, the former is a set of {\em tuple generating dependencies} and the latter consists of {\em non-conflicting keys} and {\em negative constraints}. As a result, the DL knowledge base must be restricted to specific DLs such as the $DL\text{-}Lite$ family of DLs. In normal DL logic programs, arbitrary DLs are allowed. (4) The semantics of dl-programs is defined as FLP-answer sets. In general, FLP-answer sets may not be well-supported when dealing with complex formulas, such as DL expressions other than atomic ones.

\begin{example}{\em (\cite{ShenW11})
\label{cycle}
Let $L = \emptyset$ and
$\Pi = \{A(g).~~B(g) \leftarrow C(g).~~C(g) \leftarrow ((A \sqcap \neg C) \sqcup B)(g).\}$.
Let ${\bf P} = \{A, B, C\}$, ${\bf C} = \{g\}$ and $\Omega =\{A, B, C\}$. We have $\HB_\Pi = \{A(g), B(g), C(g)\}$.
$\Pi$ has only one model consistent with $L$, $I = \{A(g), B(g), C(g)\}$. This model is not a well-supported answer set of $\KB$.
We can use a fresh DL concept $D$ to replace the DL expression $(A \sqcap \neg C) \sqcup B$ and add to $L$ an axiom $D \equiv (A \sqcap \neg C) \sqcup B$. This yields in syntax a normal dl-program of \cite{lukas-IEEE}, which is also a normal DL program in syntax.
$$
\begin{array}{ll}
L':D \equiv (A \sqcap \neg C) \sqcup B\\
\Pi': A(g).~~~B(g) \leftarrow C(g).~~~C(g) \leftarrow D(g).
\end{array}
$$
Denote this program by $\KB'$.
By adding $D$ into {\bf P} and $\Omega$ and by assuming the same {\bf C} as above, $\KB'$ has no well-supported answer set.  The semantics of disjunctive dl-programs is based on FLP-reduct,  whereby answer sets are minimal models which are not necessarily well-supported.
For $\KB'$ above, $I$ is an FLP-answer set. Observe that the evidence of the truth of $B(g), C(g), D(g)$ in the answer set can only be inferred via a self-supporting loop $B(g)\Leftarrow C(g)\Leftarrow D(g)\Leftarrow B(g)$.  Therefore, FLP-answer sets
are not necessarily the exact fixpoints of the operator $\Phi_{\KB'}$.
}
\end{example}
}

\bibliographystyle{acmtrans}
\bibliography{journal09}

\comment{
\newpage
\appendix
\section{Proofs}
\noindent
{\bf Proof of Lemma \ref{3op}:}
\begin{proof}
 We can show the following:
$$u \leq A^1(u,v) \leq A^1(u,u) = A^2(u,u) \leq A^2(u,x).$$
The first inequation is because $(u,v) \leq_p {\cal A}(u,v)$ since $(u,v)$ is ${\cal A}$-reliable. The second is due to
${\cal A}(u,v) \leq_p {\cal A}(u,u)$ because $(u,v)$ is consistent thus
$(u,v) \leq_p (u,u)$ and ${\cal A}$ is $\leq_p$-monotone.
The next equation is by the fact that ${\cal A}(u,u) = ({\cal O}(u), {\cal O}(u))$ when ${\cal A}(u,u)$ is consistent. The last inequation is due to $x \geq u$ and that ${\cal A}$ is $\leq_p$-monotone.
\end{proof}

\noindent
{\bf Notation:}
In the proofs of Lemmas \ref{AS} and \ref{rp} below,
we denote by $St_{\cal A}^i(u,v)$ ($i =1,2$) the $i$th projection of the pair $St_{\cal A} (u,v)$. We may drop ${\cal A}$ and, when no confusion arises we may
just write $St^i$ for
$St^i(u,v)$ $(i=1,2)$; this is particularly convenient for referring to individual components of $St(u,v)$.

\vspace{.1in}
\noindent
{\bf Proof of Lemma \ref{AS}:}
\begin{proof}
If either $(u,v)$ or ${\cal A}(u,u)$ is inconsistent, by definition, $St (u,v) = (\lfp({\cal A}^1(\cdot, v)), {\cal A}^2(u,v))$.
By the ${\cal A}$-prudence  of $(u,v)$, $u \leq \lfp({\cal A}^1(\cdot, v))$ and by its ${\cal A}$-reliability, $(u,v) \leq_p {\cal A}(u,v)$ hence
${\cal A}^2(u,v) \leq v$.  It follows $(u,v) \leq_p St(u,v)$.
Otherwise, since $(u,v)$ is ${\cal A}$-reliable, we have ${\cal A}^2(u,v)  \leq v$, thus
$v$ is a pre-fixpoint of ${\cal A}^2(u, \cdot)$.
Since $St^2 (u,v)$ is the least pre-fixpoint of ${\cal A}^2(u, \cdot)$, $St^2 (u,v) \leq v$. As $(u,v)$ is ${\cal A}$-prudent, it follows $u \leq St^1 (u,v)$. Putting the two together, we get $(u,v) \leq_p St (u,v)$.
\end{proof}

\vspace{.1in}
\noindent
{\bf Proof of Lemma \ref{rp}:}
\begin{proof}
We prove the assertion by considering two cases: (a) $(u,v)$ or ${\cal A}(u, u)$ is inconsistent and (b) both are consistent.

(a) By definition, we have $St (u,v)  =  (\lfp({\cal A}^1(\cdot, v)), {\cal A}^2(u, v))$.  For ${\cal A}$-reliability, we need to show
$St(u,v) \leq_p {\cal A}(St(u,v))$, which can expressed equivalently as
$(St^1, St^2) \leq_p {\cal A}(St^1, St^2)$.  By  Lemma \ref{AS}, we have  $(u,v) \leq_p (St^1,St^2)$ thus
$St^2 \leq v$.
Since ${\cal A}^1(St^1, \cdot)$ is anti-monotone and ${\cal A}$ is $\leq_p$-monotone, it follows
$$
St^1 = {\cal A}^1(St^1, v) \leq {\cal A}^1(St^1, St^2)
$$
and similarly,
$$
{\cal A}^2(St^1, St^2) \leq {\cal A}^2(u,v) = St^2
$$
Thus $(St^1, St^2)$ is ${\cal A}$-reliable.
In this case, if $(St^1, St^2)$ or ${\cal A}(St^1, St^1)$ is inconsistent, we have
$
St(St^1, St^2) = (\lfp({\cal A}^1(\cdot, St^2)), {\cal A}^2(St^1, St^2))
$.
To show ${\cal A}$-prudence, we prove the inequality below by induction.
$$
St^1 = \lfp({\cal A}^1(\cdot, v)) \leq \lfp({\cal A}^1(\cdot, St^2))
$$
Let us denote the terms in these sequences by
$x^\alpha_{1}$ and $x^\alpha_{2}$ respectively, where $\alpha$ is an ordinal.
The zero case is obvious. For any successor ordinal $\alpha + 1$, we have
$$
x_1^{\alpha+1} = {\cal A}^1(x_1^\alpha, v) \leq {\cal A}^1(x_2^\alpha, St^2) = x_2^{\alpha+1}
$$
where the inequality is by induction hypothesis. The case for limit ordinal is similar.
This proves that $St^1 \leq St^1(St^1, St^2)$ thus
$(St^1, St^2)$ is ${\cal A}$-prudent.
Otherwise, both $(St^1, St^2)$ and ${\cal A}(St^1, St^1)$ are consistent, in which case from Definition \ref{st} we get
$$
St(St^1, St^2) = (\lfp({\cal A}^1(\cdot, St^2)), \lfp({\cal A}^2(St^1, \cdot)))
$$
We can prove ${\cal A}$-prudence in the  same way as above.

(b) If $(u,v)$ and ${\cal A}(u, u)$ are both consistent, then we have
$
St^1 = {\cal A}^1(St^1,v) \leq {\cal A}^1(St^1,St^2)$ and
${\cal A}^2(St^1,St^2) \leq {\cal A}^2(u,St^2) = St^2$.
Thus the pair $(St^1, St^2)$ is ${\cal A}$-reliable. The proof of ${\cal A}$-prudence in this case is the same as that for (a).
\end{proof}

\noindent
{\bf Proof of Lemma \ref{mono1}:}
\begin{proof}
Let us show the more subtle claim on $\Phi_\KB$ extending ${\cal K}_\KB$,
 namely,
$\Phi_\KB(I, I) = ({\cal K}_\KB (I), {\cal K}_\KB (I))$, when $\Phi_\KB(I, I)$ is consistent. Note that if $(I,I)$ is inconsistent with $L$, by definition, $\Phi_\KB(I, I)$ is inconsistent.  Thus, we only need to consider the case that $(I,I)$ is consistent with $L$. Assume $\Phi_\KB(I, I)$ is consistent and show
${\cal K}_\KB (I) = \Phi^1_\KB(I, I) = \Phi^2_\KB(I, I)$.
Since $(I,I)$ is consistent with $L$, and $(I,I)$
is a 2-valued interpretation, it is easy to see that ${\cal K}_\KB (I) = \Phi^1_\KB(I, I)$, as the definition of $\Phi^1_\KB$ reduces to that of ${\cal K}_\KB$.  This is also the case for $\Phi_\KB^2$ except the extra condition,
$(I,I) \not \models_L \neg H$, in the definition.
Now, suppose this extra condition made a difference, i.e.,  some $H \in \Phi_\KB^1(I,I)$, but due to $(I,I)  \models_L \neg H$,
$H \not \in \Phi_\KB^2(I,I)$.  Then,
$\Phi_\KB(I, I) = (\Phi^1_\KB(I, I), \Phi^2_\KB(I, I))$ is an inconsistent pair, which leads to a contradiction. Thus, ${\cal K}_\KB (I) = \Phi^1_\KB(I, I) = \Phi^2_\KB(I, I)$.
\end{proof}

\noindent
{\bf Proof of Theorem \ref{theorem1}:}
\begin{proof}
Note that $\Phi_\KB^1$ of Definition \ref{phi} is the same function as $T_\KB$ of Definition \ref{7} by setting  ${\cal F}_\KB(I,J) = (\HB_\Pi, \emptyset)$ whenever $(I,J)$ is inconsistent with $L$.

Consider the sequences that construct the well-founded fixpoint of $\Phi_\KB$ (namely the least fixpoint of the stable revision operator $St_{\Phi_\KB}$) and  $\lfp(V_\KB)$, respectively:
for  all
$ i \geq 0$,
\begin{eqnarray}
\xi^0 = (\emptyset, \HB_\Pi), \ldots ,  \xi^{i+1} = St ( \xi^i),  \ldots\\
V^0 =\emptyset, \ldots , V^{i+1} = V(V^i), \ldots
\label{s}
\end{eqnarray}
where, for ease of presentation, we dropped the subscript $\Phi_\KB$ in the case of the stable revision operator $St_{\Phi_\KB}$ and the subscript $\KB$ in the case of operator $V_\KB$ (cf. Lemma \ref{V}).

Define the projection:
if $\xi^i = (x,y)$, then $(\xi^i)_1 = x$ and $(\xi^i)_2 = y$.  Let us simply write $\xi^i_1$ and $\xi^i_2$ instead.

We prove the theorem by induction on the length of the sequences, namely, $\xi^i = ((V^i)^+, PT^i)$, where $PT^i = \HB_\Pi \backslash( V^i)^-$, for all $i \geq 0$. (Intuitively, $PT^i$ is a set of potentially true atoms.)

{\em Base case:} $V^0  = \emptyset$ and $\xi^0 = (\emptyset, \HB_\Pi)$, and it is easy to check, using the definition of the $\Phi_\KB$ operator, that $\xi^0$ is $\Phi_\KB$-reliable and $\Phi_\KB$-prudent.

{\em Inductive step:}  Assume $\xi^k = ((V^k)^+, PT^k)$ and $\xi^k$ is $\Phi_\KB$-reliable and $\Phi_\KB$-prudent, for any $k \geq 0$ and show it holds for $k+1$. By Lemma \ref{mono1}, we know that $\Phi_\KB$ is an approximating operator of ${\cal K}_\KB$.  By induction hypothesis, we know that $\xi^k$ is $\Phi_\KB$-reliable and $\Phi_\KB$-prudent, and by
Lemma \ref{rp}, $\xi^{k+1} = St(\xi^k_1,\xi^k_2)$ is also $\Phi_\KB$-reliable and $\Phi_\KB$-prudent.

To prove $\xi^{k+1} = ((V^{k+1})^+, PT^{k+1})$, we consider two cases, $\xi^k$ is inconsistent with $L$ and otherwise. 

{\bf Case (i)}:  $\xi^k$ is inconsistent with $L$. In this case, $\xi^k$ is inconsistent, or not.  In the latter case,  the 2-valued interpretation $\xi^k_1$ is also inconsistent with $L$, since
$(\HB_\Pi  - \xi^k_2) \subseteq  (\HB_\Pi -  \xi^k_1)$. It follows from the definition of the operator $\Phi_\KB$ that $\Phi_\KB(\xi^k_1,\xi^k_1) = (\HB_\Pi, \emptyset)$, which  is inconsistent.  Thus, either $\xi^k$ is inconsistent or  $\Phi_\KB(\xi^k_1,\xi^k_1)$ is inconsistent.
It follows from
Definition \ref{st} that we have
\begin{eqnarray}
\xi^{k+1} = ( \lfp(\Phi_\KB^1(\cdot, \xi^k_2)), \Phi_\KB^2 (\xi^k_1,\xi^k_2)).\end{eqnarray}
Since $\Phi_\KB^1$ is the same function as $T_\KB$,
we get (cf. Lemma (\ref{V}))
\begin{eqnarray}
\xi^{k+1} = (\lfp(T_{\KB, (V^k)^-}), \Phi_\KB^2 (\xi^k_1,\xi^k_2))\end{eqnarray}
and because $\Phi_\KB^2 (\xi^k_1,\xi^k_2) =\emptyset$, we have
\begin{eqnarray}\xi^{k+1} = (\lfp(T_{\KB, (V^k)^-}), \emptyset)
\end{eqnarray}
and by definition, we get
\begin{eqnarray}\xi^{k+1} = ((V^{k+1})^+, PT^{k+1})
\end{eqnarray}

{\bf Case (ii)}: $\xi^k$ is consistent with $L$.
Since $\xi^k$ is $\Phi_\KB$-reliable,
according to the definition of stable revision operator in Definition \ref{st}, we only need to consider two subcases:  $\Phi_\KB(\xi^k_1, \xi^k_1)$ is inconsistent, and otherwise.

\vspace{.1in}
{\bf Subcase (a)}:  $\Phi_\KB(\xi^k_1, \xi^k_1)$ is inconsistent.
From Definition \ref{st}, we have
\begin{eqnarray}
\xi^{k+1} = ( \lfp(\Phi_\KB^1(\cdot, \xi^k_2)), \Phi_\KB^2 (\xi^k_1,\xi^k_2)).\end{eqnarray}
Since $\Phi_\KB^1$ is the same function as $T_\KB$, by induction hypothesis,
it is sufficient to show
\begin{eqnarray}
\HB_\Pi \backslash \Phi_\KB^2((V^k)^+, PT^k) =
U_{\KB}(V^k) \cup  Z_{\KB}(V^k).
\end{eqnarray}
By Definition \ref{phi}, for any atom $H \in \HB_\Pi$, $H \not\in \Phi_\KB^2((V^k)^+, PT^k)$ iff $((V^k)^+, PT^k) \models_L \neg H$, or both of the following hold
\begin{itemize}
\item
[(i)] $\forall I', J'$ with $(V^k)^+ \subseteq I' \subseteq J' \subseteq PT^k$ such that  $(I'J') \not \models_L H$ and,
\item
[(ii)] $\forall I', J'$ with $(V^k)^+ \subseteq I' \subseteq J' \subseteq PT^k$ such that $(I'J') \not \models_L body(r)$, for all $r \in \Pi$ with $head(r) = H$.
\end{itemize}
Then, by Definition \ref{u} and the notion of of direct negative consequences in Definition \ref{7} (the operator $Z_\KB$), it is easy to see that $H \not\in \Phi_\KB^2((V^k)^+, PT^k)$ iff $H \in U_{\KB}(V^k)$ or $H \in Z_{\KB}(V^k)$.

\vspace{.1in}
{\bf Subcase (b)}: $\Phi_\KB(\xi^k_1, \xi^k_1)$ is consistent.
By definition, we have
\begin{eqnarray}
\xi^{k+1} = ( \lfp(\Phi_\KB^1(\cdot, \xi^k_2)), \lfp(\Phi_\KB^2 (\xi^k_1, \cdot)) ).\end{eqnarray}
Again, since $\Phi_\KB^1$ is the same function as $T_\KB$, by induction hypothesis, it is sufficient to show
\begin{eqnarray}
\label{eq}
\HB_\Pi \backslash \lfp(\Phi_\KB^2 ((V^{k})^+, \cdot)) =
U_{\KB}(V^k) \cup  Z_{\KB}(V^k).
\end{eqnarray}

Note that $Z_{\KB}(V^k) = \HB_\Pi \backslash R^k$,
where
$R^k = \{H \in \HB_\Pi~|~
((V^k)^+, PT^k) \not \models_L \neg H)\}$, as defined in $\Phi_\KB^2$ of Definition \ref{phi}.

 Notice also that, according to
Definition \ref{phi},  for any $(I,J)$, $H \in \Phi_\KB^2(I,J)$ iff $(I,J) \not \models_L \neg H$ and either condition (a) or condition (b) holds. Similar to Lemma \ref{V}, it is easy to show that
$$\lfp(\Phi_\KB^2 ((V^k)^+, \cdot)) = \lfp(\Phi_\KB'^2 ((V^k)^+, \cdot)) \cap R^k$$ where $\Phi_\KB'^2$ is the same as $\Phi_\KB^2$ except that we do not check the condition $((V^k)^+, PT^k) \not \models_L \neg H$ inside the fixpoint construction.  Since $Z_\KB(V^k) = \HB_\Pi\backslash R^k$,
it follows  that to prove
equation (\ref{eq}), it suffices to show that
for all
$H \in \HB_\Pi$, $H \not \in \lfp(\Phi_\KB'^2 ((V^k)^+, \cdot))$ iff
$H \in U_{\KB}(V^k)$.\footnote{Intuitively, this says that the atom set that is complement to
the set of atoms that are potentially true is precisely the greatest unfounded set.}

Below, let us denote by $\{y_n\}$ the sequence of constructing $\lfp(\Phi_\KB'^2 ((V^k)^+, \cdot))$, and by  $y_\infty$ its least fixpoint.

($\Leftarrow$) By contraposition. Assume $H \in y_\infty$ and show $H \not \in U_{\KB}(V^k)$.
By definition, $H \in y_\infty$ iff $H \in y_i$ for some term in $\{y_n\}$  iff (i) $H \in (V^k)^+$ or
(ii)
$(I', J') \models_L H$ for some $(I',J')$ s.t. $(V^k)^+ \subseteq I'\subseteq J' \subseteq  y_i$ or (iii)
$(I',J') \models_L body(r)$ for some rule
$r \in \Pi$ with $head(r) = H$ and $(I',J')$ s.t. $ (V^k)^+ \subseteq I'\subseteq J' \subseteq  y_i$, .  By induction on the sequence $\{y_n\}$, it is easy to show that
$H \not \in  U_{\KB}(V^k)$, since the statement in (i-iii) above violates either condition (a) or condition (b) in
Definition \ref{u}.

($\Rightarrow$) Assume $H \not \in y_\infty$ and show $H \in U_{\KB}(V^k)$.
$H$ either appears as the head of a rule, or not.  In the latter case, $H \in U_{\KB}(V^0)$ and hence $H \in U_{\KB}(V^k)$. In the former case, by definition,
that $H \not \in y_\infty$ implies that for every rule $r \in \Pi$ with $head(r) = H$, and
$\forall I',J'$  s.t. $ (V^k)^+ \subseteq I' \subseteq J' \subseteq y_\infty, (I',J') \not \models_L body(r)$. It follows from definition that $\HB_\Pi \backslash y_\infty$ is an unfounded set relative to $V^k$. Since $H \in \HB_\Pi \backslash y_\infty$, it is in the greatest unfounded set, $U_{\KB}(V^k)$,
 relative to $V^k$.
\end{proof}

\noindent
{\bf Proof of Lemma \ref{2-3}:}
\begin{proof}
If a literal $A$ in $C$ is positive,
we show that
$$
\begin{array}{ll}
\mbox{ (i) $(I,J) \models_L A$ iff (ii) $\langle I_1,I_2 \rangle \models_L A$ iff (iii) $\ll  \!\! I_1,I_2 \!\!\gg \models_L A$.}
\end{array}
$$
That (i) implies (ii) is by definition. That (ii)  implies (iii) is because every total interpretation is a partial one.  From (iii), for every 2-valued interpretation $F$, $I_1 \subseteq F \subseteq I_2$, we have $ F \models_L A$. It follows
$(I_1,I_1) \models_L A$ and $(I_2,I_2) \models_L A$, and since
$I_1 \subseteq I_2$, we have  $I_1 \cup \neg . \bar{I_2} \cup \neg . (I_2-I_1) \models_L A$ and $I_1 \cup \neg . \bar{I_2} \cup (I_2-I_1) \models_L A$ (i.e., it does not matter whether $(I_2-I_1)$ serves as a positive premise or a  negative premise), which implies $(I_1,I_2) \models_L A$. These equivalences imply
$\langle I_1,I_2 \rangle \not \models_L A$ iff $\ll  \!\! I_1,I_2 \!\!\gg \not \models_L A$, for any negative literal $not~A$ in $C$, where the former equals $\langle I_1,I_2 \rangle \models_L not~A$ and the latter
$\ll  \!\! I_1,I_2 \!\!\gg  \models_L not~A$.
Thus the claim holds.
\end{proof}

\vspace{.1in}
\noindent
{\bf Proof of Theorem \ref{theorem2}:}
\begin{proof}
Consider the fixpoint construction by the operators ${\cal T}_\KB(\cdot, I)$ and  ${\Phi^1}_\KB(\cdot, I)$. Let us use a short notation for the respective sequences as
\begin{eqnarray}
\label{9}
{\cal T}_\KB^0 ( = \emptyset), \ldots  , {\cal T}_\KB^{k+1} (= {\cal T}_\KB ({\cal T}_\KB^k, I)), \ldots ~~~~~\\
\label{10}
{\Phi^1}_\KB^0 (=\emptyset), \ldots , {\Phi^1}_\KB^{k+1}  (= \Phi^1({\Phi^1}_\KB^k, I)), \ldots,
\end{eqnarray}
for all $k \geq 0$. Define $E_0 = \emptyset$ and
$E_i = \{H ~|~({\cal T}^{i-1}_\KB, I) \models_L H\}$ for all $i \geq 1$. We show ${\Phi^1}^i_\KB = E_i \cup {\cal T}^{i}_\KB$, for all $i \geq 0$. The base case is trivial. For the inductive step, assume for any $k \geq 0$ the equation holds and we show it holds for $k+1$. By definition and monotonicity of the operator ${\cal T}_\KB$, $({\cal T}^{i}_\KB, I) \models_L E_i$, and it follows from the first-order entailment that
 for any atom $H \in \HB_\Pi$,
\begin{eqnarray}
\label{e1}
({\cal T}^{i}_\KB, I) \models_L H \mbox{ iff } (E_i\cup {\cal T}^{i}_\KB , I) \models_L H,
\end{eqnarray}
and similarly
\begin{eqnarray}
\label{e2}
\langle {\cal T}^{i}_\KB, I \rangle \models_L H \mbox{ iff }
\langle E_i \cup {\cal T}^{i}_\KB, I\rangle \models_L H.
\end{eqnarray}
We have
$$
\begin{array}{ll}
{\Phi^1}^{k+1}_\KB = \{H \in \HB_\Pi | ({\Phi^1}^k_\KB, I) \models_L H\} \cup
\{ head(r)~|~ r \in \Pi, \langle {\Phi^1}^k_\KB,I\rangle \models_L body(r)\}\\
~~~~~~~~~~~~~~~~~~~~~~~~~~~~~~~~~~~~~~~~~~~~~~~~~~~~~~~~~~~~~~~~~~~~~~~~~~~~~~~~~~~~~~~~~~~~~~~~~~~~~~~~~~~~~~~~~~~~~~~~~~~~~~(\mbox{by Definition \ref{phi}}) \\
=\{H \in \HB_\Pi | (E_k \cup {\cal T}^k_\KB, I) \models_L H\} \cup
\{head(r)~|~ r \in \Pi, \langle E_k\cup {\cal T}^k_\KB,I\rangle \models_L body(r)\}\\
~~~~~~~~~~~~~~~~~~~~~~~~~~~~~~~~~~~~~~~~~~~~~~~~~~~~~~~~~~~~~~~~~~~~~~~~~~~~~~~~~~~~~~~~~~~~~~~~~~~~~~~~~~~~~~~~~(\mbox{by  induction hypothesis})\\
= \{H \in \HB_\Pi | ({\cal T}^k_\KB, I) \models_L H\} \cup  \{head(r)~|~ r \in \Pi, \langle {\cal T}^k_\KB, I\rangle \models_L body(r)\}\\
~~~~~~~~~~~~~~~~~~~~~~~~~~~~~~~~~~~~~~~~~~~~~~~~~~~~~~~~~~~~~~~~~~~~~~~~~~~~~~~~~~~~~~~~~~~~~~~~~~~~~~~~~~~~~~~~~~~~~~~(\mbox{by (\ref{e1}) and (\ref{e2})})\\
=\{H \in \HB_\Pi | ({\cal T}^k_\KB, I) \models_L H\} \cup
\{head(r)~|~ r \in \Pi, \ll \!\!  {\cal T}^k_\KB, I  \!\!\gg \models_L body(r)\}\\
~~~~~~~~~~~~~~~~~~~~~~~~~~~~~~~~~~~~~~~~~~~~~~~~~~~~~~~~~~~~~~~~~~~~~~~~~~~~~~~~~~~~~~~~~~~~~~~~~~~~~~~~~~~~~~~~~~~~~~~~~~~~~~~~~~~(\mbox{by  Lemma \ref{2-3}})\\
= E_{k+1} \cup {\cal T}^{k+1}_\KB ~~~~~~~~~~~~~~~~~~~~~~~~~~~~~~~~~~~~~~~~~~~~~~~~~~~~~~~~~~~~~~~~~~~~~~~~~~~~~~~~~~~~~~~~~~~~~~~~~~~~~~~~(\mbox{by definition})
\end{array}
$$
Let ${\Phi^1}_\KB^\alpha$ and ${\cal T}_\KB^\alpha$ be the respective fixpoints. 
We then have
\begin{eqnarray}
\label{e3}
{\phi^1}_\KB^\alpha  \!=\! \{H \!\in \!\HB_\Pi | ({\cal T}^\alpha_\KB, I) \models_L \!\! H\} \cup
\{head(r)~| ~r \in \Pi,  \ll\!\!  {\cal T}^\alpha_\KB, I  \!\!\gg \models_L  body(r)\}
\end{eqnarray}
Suppose $(I,I)$ is a well-supported answer set of $\KB$. By definition, this means
\begin{eqnarray}
\label{e4}
I =\! \{H \!\in \!\HB_\Pi | ({\cal T}^\alpha_\KB, I) \models_L \!\! H\} \cup
\{head(r)~| ~r \in \Pi,  \ll\!\!  {\cal T}^\alpha_\KB, I  \!\!\gg \models_L  body(r)\}
\end{eqnarray}
We then have ${\Phi^1}_\KB^\alpha = I$. Since $\Phi^1_\KB(I,I) = I$ by definition,
it is easy to see that
$\Phi^2_\KB(I,I) = I$.  Thus,
$(I,I)$ is a fixpoint of $St_{\cal A}$, thus a
stable fixpoint of $\Phi_\KB$.

Conversely,
suppose $(I,I)$ is a stable fixpoint of ${\cal A}$. Then $\Phi^1_\KB(I,I) = I$.
Together with equation (\ref{e3}), we derive equation (\ref{e4}). As  $I$ is a model of $\KB$, it follows that it is
a well-supported answer set of $\KB$.
This completes the proof.
\end{proof}
}

\end{document}